\documentclass[letterpaper, 10 pt, conference]{ieeeconf}  % Comment this line out if you need a4paper

\IEEEoverridecommandlockouts                              % This command is only needed if 
\usepackage{multicol}
\usepackage{algorithmic}
\usepackage{algorithm}
\usepackage{arydshln,leftidx,mathtools}
\usepackage{amsmath}
\usepackage{amssymb}
\usepackage[T1]{fontenc}
\usepackage[latin1]{inputenc} 
\usepackage{color}
\usepackage{graphicx}
\usepackage{wrapfig}
\usepackage{color}
\usepackage{lipsum}
\usepackage[usenames,dvipsnames]{xcolor}
\graphicspath{ {./images/} }
\newcommand{\MORPH}{\mathcal{T}}
\newcommand{\QNUM}{n_\mathrm{p}}

\newcommand{\IM}{\mathrm{Im}}
\newcommand{\ALPV}{ALPV\ }
\newcommand{\LFPV}{LPV-LFR\ }
\newcommand{\ALPVs}{ALPVs\ }

\newcommand{\Rank}{\mathrm{rank}}

\newtheorem{Theorem}{\textbf{Theorem}}

\newtheorem{Corollary}{\textbf{Corollary}}

\newtheorem{Definition}{\textbf{Definition}}
\newtheorem{Remark}{\textbf{Remark}}
\newtheorem{Notation}{\textbf{Notation}}
\newtheorem{Example}{\textbf{Example}}

\newcommand{\NX}{n_\mathrm{x}}
\newcommand{\NY}{n_\mathrm{y}}
\newcommand{\NU}{n_\mathrm{u}}
\newcommand{\NP}{n_\mathrm{p}}
\newcommand{\NTH}{n_\mathrm{\theta}}

\newcommand{\INP}{\{0,\ldots, n_\mathrm{p}\}}

\newcommand{\rank}{\mathrm{rank}}

\newcommand{\diag}[1]{%
	{\textup{diag}}{#1}}

\title{\LARGE \bf
   Structural properties of LPV to LFR transformation: minimality, input-output behavior and identifiability.
}

\author{Ziad Alkhoury$^{1,2}$ and Mih\'aly Petreczky$^{3}$ and Guillaume Mercère$^{1}$% <-this % stops a space
\thanks{$^{1}$
	Ziad Alkhoury and Guillaume Mercère are with University of Poitiers, Laboratoire d'Informatique et d'Automatique pour les Systèmes, Bâtiment B25, 2 rue Pierre Brousse, TSA 41105, 86073 Poitiers Cedex 9. {\tt\small ziad.alkhoury@univ-poitiers.fr, guillaume.mercere@univ-poitiers.fr}}%
\thanks{$^{2}$
	Ziad Alkhoury is also with Ecole des Mines de Douai, F-59500 Douai, France.}%
\thanks{$^{3}$Mih\'aly Petreczky is with the CNRS, Centrale Lille, UMR 9189 - CRIStAL-Centre de Recherche en Informatique, Signal et Automatique de Lille, F-59000 Lille, France, {\tt\small mihaly.petreczky@ec-lille.fr}}%
\thanks{*This work was partially supported by ESTIREZ project of Region Nord-Pas de Calais, France.}% <-this % stops a space
}

\begin{document}

\maketitle
\thispagestyle{empty}
\pagestyle{empty}

%%%%%%%%%%%%%%%%%%%%%%%%%%%%%%%%%%%%%%%%%%%%%%%%%%%%%%%%%%%%%%%%%%%%%%%%%%%%%%%%
\begin{abstract}
	In this paper, we introduce and study important properties of the transformation of Affine Linear Parameter-Varying (ALPV) state-space representations into Linear Fractional Representations (LFR). More precisely, we show that $(i)$ state minimal ALPV representations yield minimal LFRs, and vice versa, $(ii)$ the input-output behavior of the ALPV represention determines uniquely the input-output behavior of the resulting LFR, $(iii)$ structurally identifiable ALPVs yield structurally identifiable LFRs, and vice versa. We then characterize LFR models which correspond to equivalent ALPV models based on their input-output maps. As illustrated all along the paper, these results have important  consequences for identification and control of systems described by LFRs.
\end{abstract}

%%%%%%%%%%%%%%%%%%%%%%%%%%%%%%%%%%%%%%%%%%%%%%%%%%%%%%%%%%%%%%%%%%%%%%%%%%%%%%%%
\section{INTRODUCTION}

Linear Fractional Transformation (LFT) is one of the main tools used in the past decades for studying uncertain systems (see, \emph{e.g.}, \cite{Zhou1996}). For instance, Linear Fractional Representations (LFRs) have been widely used in control synthesis of Liner Parameter Varying (LPV) systems or for $H_\infty$ optimal control (see \cite{Zhou1996,Rugh2000,BriatBook} for overview). More recently, the LFRs have attracted a lot of attention as far as system identification is concerned, (see \cite{LP96, LP97, LP99, CL08, Vizer2013b}). For LPV model-based controller design, several solutions first consist in transforming the LPV system into an LFR. This step indeed allows us to use control tools developed for LFRs to design a controller to guarantee the satisfactory closed loop operation of the LPV plant in many operating conditions. As far as system identification is considered, it is clear from the literature (see, among others, \cite{Tot10, Lal11, Lov14}) that the branch of data-driven modeling dedicated to LPV model identification is a lot more mature than the one dealing directly with LFRs. These observations mean that, in control design as well as in system identification, LPV models are often used as an intermediary representation, whose main purpose is to serve as a source for an LFT description. It is thus of prime interest in control in general to study the transformation of LPV to LFR closely. 

In this paper, a specific attention is paid to important realization theory concepts like minimality and input-output equivalence of model representations. The reason why this last point (\emph{i.e.}, input-output equivalence of LPV models and LFRs) is a crucial and a challenging problem in system identification and controller design can be illustrated as follows. Input-output equivalence of two LPV models means that these two models yield the same outputs for the all inputs and scheduling signals. Input-output equivalence of the corresponding LFRs means that they both yield the same outputs for the all input and \textbf{all} the choice of uncertainty block $\Delta$. The latter point (\emph{i.e.}, for all $\Delta$) leads us to the conclusion that the LFRs should behave the same way for $\Delta$ blocks which do not arise from scheduling variables. The uncertainty operator in $\Delta$ could be, for instance, any stable non-rational transfer function.  In fact, we can not even conclude that two LFRs which arise from two input-output equivalent LPV models can be interconnected\footnote{Note that the interconnection of an LFR with a block $\Delta$ need not always be well-defined.} with the same uncertainty block $\Delta$. This observation is not an issue if the LPV model is known from first principles. However, if the LPV model is identified from data, leading to a black-box model, then different identification methods applied to the same measurements may yield different LPV models which are, at most, input-output equivalent. By keeping in mind that the existence of a controller for an LFR only depends on its input-output behavior (while the disturbances block $\Delta$ is assumed to have a bounded norm), the situation described above means that the outcome of controller synthesis may depend on the choice of the identification method, even under ideal conditions. This is clearly an undesirable situation.

%\red{Let us further illustrate why it is an important and a challenging problem to study the LPV to LFR transformation. Assuming the availability of a reliable and sufficiently large data-sets, it is now well-known}

This simple illustration of system identification for control clearly points out the fact that we need to understand the relationship between the input-output behavior of LPV models and the corresponding LFRs. Indeed, the measurements allow us to say something about the input-output behavior of the underlying LFR for those choices of the uncertainty block which correspond to scheduling variables. However, it is not \emph{a priori} clear that this information is sufficient to determine the input-output behavior of the LFR for all
other choices of the uncertainty structure.

In this paper, we first show that the transformation from ALPV models to LFRs preserves minimality. This enables
us to show in a second step that the input-output behavior of an ALPV model uniquely determines the input-output behavior of the corresponding LFR. Indeed, from \cite{Alkhoury15}, it follows that minimal ALPV models with the same input-output behavior are related by a constant state isomorphism. We then show that that the classical transformation from ALPV models to LFR models preserves isomorphism. Hence, not only input-output equivalent ALPV models yield input-output equivalent LFRs, but minimal and input-output equivalent LPV models yield isomorphic LFRs. This result has an interesting consequence as far as controller design is concerned. If attention is restricted to minimal models, then control synthesis does not depend on which representative of the class of input-output equivalent ALPV we have picked to design the controller. We can thus conjecture that there is no advantage in using
non-minimal ALPV or LFT models for control synthesis. This is the case for LTI systems. The proof of this conjecture remains future work though. 

%Minimality of LFR models is already related to the realization theory of Formal Power Series representation (FPS) \cite{Berstel1988}. It is good to note that formal power series were used in systems theory earlier and are also used in many branches of mathematics. In \cite{Beck2001}, a method is presented to obtain a minimal LFR realization of a minimal FPS realization, and vice versa. In \cite{Petreczky2005}, the realization theory of FPS is used to derive a characterization of Bilinear Switched realizations.

\textbf{Related work}
To the best of our knowledge, the results of the paper are new. While the idea of using LFRs for LPV control is a standard one \cite{BriatBook,Rugh2000,Scherer2001,PACKARD1994}, and the transformation was described before \cite{Verdult2002,Vizer2013b,Casella2008}, the structural properties of this transformation, such as preservation of minimality and input-output equivalence were not investigated before. Notice that, in deriving the results of the paper, we use realization theory of ALPV systems \cite{PTM15,PM12}, and realization theory of LFRs (viewed as multidimensional systems) \cite{Beck2001,Ball2005}.

\textbf{Outline of the paper}
In Section~\ref{sec:FLPV}, we present formal definitions to setup the framework of this paper. 
 Section~\ref{sec:Connections_minimality_LPV_LFPV} contains the main results dedicated to the connection between ALPV models and the corresponding \LFPV ones. 
% In Section~\ref{sec:ALPV_FPS}, we present the relationship between \ALPV models, and equivalent FPS representations, while in Section~\ref{sec:LFPV_ALPV}, we characterize LFR models that correspond to equivalent \ALPV ones, based on their series representations. 
Finally, Section~\ref{sec:conclusion} concludes the paper.

\section{The formal setup: LFR and ALPV models} \label{sec:FLPV}
 In this section we present the formal setup of the problem considered in this paper. First, in subsections~\ref{sec:LFR} and~\ref{sec:ALPV}, we define \LFPV and \ALPV models, respectively. In subsection~\ref{sec:LPV_LFR_Transformation}, the transformation from \LFPV to \ALPV models and vice versa is presented. Motivating examples are presented in subsection~\ref{problem:form}, while in subsection~\ref{sec:problem_formulation} the studied problems are formulated.
% That is, we review the definition of LFR models in general, and that of LPV-LFR models in 
% particular.  The latter is a subclass of LFR models which arise from ALPV models.
% Furthermore, we review the definition of ALPV models. We also present an overview of the relevant results on realization theory of LFR and
% ALPV models. 
%We define formally what we mean by the input-output behavior of such models.
% Note that the latter is not obvious, as LFR models need not be well-posed, \emph{i.e.}, they need not have a 
% well-defined solution. This is not surprising, since LFRs represent a feedback interconnection of an
% LTI system with an uncertainty block, and it is well-known that feedback interconnections need not
% be well-posed \cite{Beck1999,ZD98}. 

\subsection{General LFR models as multidimensional systems}\label{sec:LFR}
A \emph{linear fractional representation} (abbreviated as LFR) is presented in Figure \ref{fig:lft_figure}, where 
$M$ is a tuple of matrices $(A,B,C,D)$ representing an LTI system, and 
%, \[
%M     M=\begin{bmatrix} A & B \\ C & D \end{bmatrix},
%\]
%$A \in \mathbb{R}^{n \times n}$, $B \in \mathbb{R}^{n \times m}$, $C \in \mathbb{R}^{p \times n}$, $D \in \mathbb{R}^{p \times m}$, and 
$\Delta$ is a linear operator on suitable function spaces. 
Note the feedback loop in Figure \ref{fig:lft_figure} is not necessarily well-posed. Hence, in order to define the
input-output behavior of an LFR formally,  we have to impose additional
conditions on $M$ and $\Delta$.  There are several such conditions, and their relationship is not trivial \cite{Beck1999,ZD98}. In order to avoid to deal with this issue, we will define \emph{formal input-output maps}, by viewing LFRs as multidimensional systems
\cite{Ball2005}.  We will see that the formal input-output map determines the input-output behavior of the LFR, for those cases which
are of interest for this paper and for which the interconnection is well-posed. 
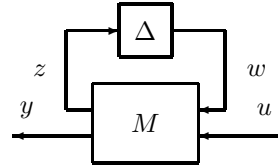
\begin{figure}%{0.2\textwidth}
	\centering
	\begin{picture}(120,70)(0,0)
	\linethickness{0.3mm}
	\put(30,30){\line(1,0){40}}
	\put(30,0){\line(0,1){30}}
	\put(70,0){\line(0,1){30}}
	\put(30,0){\line(1,0){40}}
	\put(50,15){\makebox(0,0)[cc]{$M$}}
	\put(50,50){\makebox(0,0)[cc]{$\Delta$}}
	\put(92,35){\makebox(0,0)[cc]{$w$}}
	\put(10,35){\makebox(0,0)[cc]{$z$}}
	\put(5,20){\makebox(0,0)[cc]{$y$}}
	\put(95,20){\makebox(0,0)[cc]{$u$}}
	\linethickness{0.3mm}
	\put(40,60){\line(1,0){20}}
	\put(40,40){\line(0,1){20}}
	\put(60,40){\line(0,1){20}}
	\put(40,40){\line(1,0){20}}
	\linethickness{0.3mm}
	\put(70,10){\line(1,0){30}}
	\put(70,10){\vector(-1,0){0.12}}
	\linethickness{0.3mm}
	\put(0,10){\line(1,0){30}}
	\put(0,10){\vector(-1,0){0.12}}
	\linethickness{0.3mm}
	\put(20,50){\line(1,0){20}}
	\put(40,50){\vector(1,0){0.12}}
	\linethickness{0.3mm}
	\put(20,20){\line(0,1){30}}
	\linethickness{0.3mm}
	\put(20,20){\line(1,0){10}}
	\linethickness{0.3mm}
	\put(60,50){\line(1,0){20}}
	\linethickness{0.3mm}
	\put(80,20){\line(0,1){30}}
	\linethickness{0.3mm}
	\put(70,20){\line(1,0){10}}
	\put(70,20){\vector(-1,0){0.12}}
	\end{picture}
	\caption{General LFR}
	\label{fig:lft_figure}
\end{figure}
\begin{Definition}%[LFR as a multidimensional system]
\label{LFR:def}
 An LFR is a tuple 
 \begin{equation}
 \label{LFRform}
   \mathcal{M}=(p,m,d, \{n_i\}_{i=1}^{d},A,B,C,D),
 \end{equation}
  where $p,m,d$, $n_{i}$, $i \in \{1,\ldots,d\}$
 are positive integers, and  $A \in \mathbb{R}^{n \times n}$, $B \in \mathbb{R}^{n \times m}$, $C \in \mathbb{R}^{p \times n}$ and $D \in \mathbb{R}^{p \times m}$
 are matrices, where $n=\sum_{i=1}^{d} n_i$. 
\end{Definition}
In Definition \ref{LFR:def}, the choice of integers $\{n_i\}_{i=1}^d$ expresses the tacit assumption that $\Delta=\diag{[\delta_1I_{n_1}, \dots, \delta_dI_{n_d}]}$ when defining the behavior of $\mathcal{M}$, where $\delta_i, \forall i=1,...,d,$ are linear operators on scalar valued sequences, and $I_{n_i}$ is the $n_i \times n_i$ identity matrix. 
Here we used the standard notation of \cite{Beck1999,ZD98}. That is, if $\delta$ is a linear operator on scalar valued sequences, then 
 $\delta I_n$ stands for the linear operator on sequences with values from $\mathbb{R}^n$, such that the result of applying $\delta I_n$ to a sequence is obtained by applying 
$\delta$ to each coordinate of the sequence $x$, see \cite{Beck1999,ZD98} for the formal definition.  
Next, we define what we mean by a formal input-output map of an LFR. To this end, we need the following notation.
\begin{Notation}[Free monoid $\mathcal{X}^{*}$] Let $\mathcal{X}^*$ be the monoid generated by a nonempty finite set $\mathcal{X}$. An element $w \in \mathcal{X}^*$ of length $|w|=n$, is a sequence of the form $x_1 x_2 \dots x_n$, where $x_i \in \mathcal{X}, \forall i=1,2,...,n$. Denote $\epsilon$ the empty sequence where $|\epsilon|=0$.
\end{Notation}
\begin{Definition}[Canonical partitioning]
Let $\mathcal{M}$ be an LFR of the form \eqref{LFRform}. The collection $\{H_i, F_{i,j},G_j\}_{i,j=1}^{d}$ of matrices such that
 $F_{i,j} \in \mathbb{R}^{n_i \times n_j}$, $G_j \in \mathbb{R}^{n_j \times m}$, $H_i \in \mathbb{R}^{p \times n_i}$, $i,j=1,\ldots,d$ and
   \[
      \begin{split}
        & A=\begin{bmatrix} 
              F_{1,1} & F_{1,2} & \cdots & F_{1,d} \\
              F_{2,1} & F_{2,2} & \cdots & F_{2,d} \\
              \vdots & \vdots   & \ddots & \vdots \\
              F_{d,1} & F_{d,2} & \cdots & F_{d,d} 
          \end{bmatrix}, \quad  B=\begin{bmatrix} G_1 \\ G_2 \\ \vdots \\ G_d \end{bmatrix}, \\
      & C=\begin{bmatrix} H_1 & H_2 & \cdots & H_d \end{bmatrix},
      \end{split}
    \]
is called the \emph{canonical partitioning} of $\mathcal{M}$. 
\end{Definition}
Notice that $D \in  \mathbb{R}^{p \times m}$ is not subject to partitioning according to $i$ and $j$.
\begin{Definition}\label{def:fiom_lfr}%[Formal input-output map of an LFR]
  The \emph{formal input-output map} of an LFR $\mathcal{M}$ is a function
%  $X_{\mathcal{M}}:\mathcal{X}^{*} \rightarrow \mathbb{R}^{n \times m}$ and 
  $Y_{\mathcal{M}}:\mathcal{X}^{*} \rightarrow \mathbb{R}^{p \times m}$, 
  $\mathcal{X}=\{1,\ldots,d\}$, and it is defined as follows: $Y_{\mathcal{M}}(\epsilon)=D$ and, for all
   $s=i_1\cdots i_k$, $i_1,\ldots,i_k \in \mathcal{X}$, $k >  0$,
  \[
     \begin{split}
     & Y_{\mathcal{M}}(s)=\left\{\begin{array}{rl} H_{i_1}G_{i_1}  &  k = 1, \\
                                                    H_{i_k}F_{i_k,i_{k-1}} \cdots F_{i_2,i_1}G_{i_1} & k > 1 ,
                                   \end{array}\right. \\
     \end{split}
  \]
   where    $\{H_i, F_{i,j},G_j\}_{i,j=1}^{d}$ is the canonical partitioning of $\mathcal{M}$.%   and \red{ $E_{i}=\begin{bmatrix} O_{n_1+\cdots n_{i-1},n_i}^\top &  I_{n_i} & O_{n_{i+1}+\cdots+n_d,n_i}^\top \end{bmatrix}^\top$} and  $O_{k,l}$ denotes  a $k \times l$   matrix with all zero elements. 
\end{Definition}
%  \cite{Berstel1988}, ${Y}_{\mathcal{M}}$ is a  formal power series. 
%
 \begin{Remark}\emph{(Formal input-output map and star product $\mathcal{M} \star \Delta$)}
 \label{rem:beck}
 Below we explain the relationship between $Y_{\mathcal{M}}$ and the usual 
 star product $\mathcal{M} \star \Delta$. We use the notation and terminology of \cite{Beck1999}. In particular, we denote by
 $l_2(\mathbb{R}^{n})$ the space of all $l_2$ sequences taking values in $\mathbb{R}^{n}$ and we denote by 
 $\mathcal{L}(l_2)$ the set of all bounded operators from $l_2(\mathbb{R})$ to $l_2(\mathbb{R})$, and we use $\|.\|$ to denote the induced operator norm of
 an operator from $\mathcal{L}({l_2})$.  
%The motivation for this particular definition lies in the result of \cite{Beck1999}, which explains the relationship between  formal input-to-state and true input-output maps of the LFR obtained by Star Product and their system-theoretically meaningful counterparts. To this end, we will need the following notation. \red{DO WE NEED ALL?}
% \begin{Notation}
%  Denote by $l_2(\mathbb{R}^{n})$ the set of all sequence $z:\mathbb{N} \rightarrow \mathbb{R}^{n}$ such that the infinite sum $\sum_{k=0}^{\infty} \| z(k)\|^2_2$ is
%  convergent. It is well-known that $l_2(\mathbb{R}^n)$ forms an Hilbert-space with the scalar product 
%  $<z,s> = \sum_{k=0}^{\infty} z(k)^Ts(k)$ and norm $\|z\|=\sqrt{\sum_{k=0}^{\infty} \|z(k)\|_2^2}$. 
%  Denote by $\mathcal{L}(l_2)$ the set of all bounded linear operators defined on $l_2(\mathbb{R})$. 
%  If $z \in l_2(\mathbb{R}^n)$, then the coordinates $z_1,\ldots,z_n$ of $z$ are in $l_2(\mathbb{R})$.
%  If $A \in \mathbb{R}^{n \times m}$ and
%  $\delta \in \mathcal{L}(l_2)$, then $\delta A$ is interpreted as a bounded linear operator from
%  $l_2(\mathbb{R}^m)$ to $l_2(\mathbb{R}^n)$ such that if $z \in l_2(\mathbb{R}^m)$ then the $i$th coordinate of
%  $(\delta A)(z)$ equals $\delta(\sum_{j=1}^m A_{i,j}z_j)$, $i=1,\ldots,n$.
%  Finally, if $\delta \in \mathcal{L}(l_2)$, we denote by $\|\delta\|$ the induced norm of the operator $\delta$. 
% \end{Notation}
  Following \cite{Beck1999}, for any $\gamma > 0$ we define the set
\[
\begin{split}
	& \mathbf{\Delta}_{\gamma} =  \{\Delta = \diag{[\delta_1I_{n_1}, \dots, \delta_dI_{n_d}]}: \delta_i \in \mathcal{L}(l_2), \\
        &  \|\delta_i\| \le \gamma, i=1,\ldots,d \}.
\end{split}
\]
%A sufficient condition for the product $\mathcal{M} \star \Delta$ to be well defined for all blocks $\Delta \in \mathbf{\Delta}$ is that
%$\mathcal{M}$ is stable, see \cite{Beck1999,ZD98}, where
%red{WE NEED A SENTENCE TO INTRODUCE STABILITY AND WHY WE NEED IT}
An LFR $\mathcal{M}$ of the form \eqref{LFRform} is called $\gamma$-\emph{stable}, if $(I_n-A\Delta)$ is an invertible bounded linear operator on $l_2(\mathbb{R}^n)$ for all $\Delta \in \mathbf{\Delta}_{\gamma}$. % From \cite{Beck1999} it follows that this is equivalent to existence of positive definite matrices $Y_i \in \mathbb{R}^{n_i \times n_i}$, $i=1,\ldots,d$, such that for $Y=\diag{[Y_1,\dots,Y_d]}$, $AYA^T - Y < 0$. 
Recall from \cite{Beck1999,ZD98} that if  $\mathcal{M}$ is $\gamma$-stable, then 
for any input $u \in l^2(\mathbb{R}^m)$ and uncertainty block $\Delta \in \mathbf{\Delta}_{\gamma}$, 
the feedback interconnection on Fig. \ref{fig:lft_figure} is well defined, and the corresponding output 
$y \in l^2(\mathbb{R}^p)$ 
satisfies $y=(\mathcal{M} \star \Delta)u$, where the bounded linear operator 
$(\mathcal{M} \star \Delta):l^2(\mathbb{R}^m) \rightarrow l^2(\mathbb{R}^p)$ is defined by
 $\mathcal{M} \star \Delta:=D+C\Delta(I_n-A\Delta)^{-1}B$.
From \cite{Beck1999}, it follows that, if $\mathcal{M}$ is a $\gamma$-stable LFR, then 
   \[
      \begin{split}
    %   & Delta)=\sum_{k=1}^{\infty} \sum_{j_1,\ldots,j_k=1}^{d} X_{\mathcal{M}}(i_1\cdots i_k)(\delta_{i_1}\delta_{i_2} \cdots \delta_{i_k})I_m u, 
      %&  
      & (\mathcal{M} \star \Delta)u\!\!=\!\! \sum_{s \in \mathcal{X}^{*}} Y_{\mathcal{M}}\delta_sI_m u,
%Y_{\mathcal{M}}(\epsilon)u+ 
       %\sum_{k=1}^{\infty}\sum_{j_1,\ldots,j_k=1}^{d}\!\!\! Y_{\mathcal{M}}(i_1\cdots i_k)(\delta_{i_k}\delta_{i_{k-1}}\!\!\cdots\delta_{i_1}\!\!)I_m u,  \\
      \end{split}
   \]
for all $\Delta \in \mathbf{\Delta}_{\gamma}$, 
   where $\delta_{\epsilon}$ is the identity operator and for $s=i_1\cdots i_k$, $i_1,\ldots,i_k \in \mathcal{X}$, $k> 0$,
 %$\delta_{i_k} \cdots \delta_{i_1}$ is the composition of the operators in that order:
 % \emph{i.e.} 
  $\delta_{s}(z)=\delta_{i_k}(\delta_{i_{k-1}}(\cdots \delta_{i_1}(z))\cdots)$ for all
   $z \in l_2(\mathbb{R})$. 
   That is, \emph{the formal input-output map $Y_{\mathcal{M}}$ determines the star product $\mathcal{M} \star \Delta$ uniquely for
   all $\Delta \in \mathbf{\Delta}$, provided $\mathcal{M}$ is stable}.
   %In particular, if the formal input-output maps of two stable LFRs $\mathcal{M}_1$ and $\mathcal{M}_2$ coincide, i.e. 
   %$Y_{\mathcal{M}_1}=Y_{\mathcal{M}_2}$,  then for all 
   %$\Delta \in \mathbf{\Delta}$, $\mathcal{M}_1 \star \Delta=\mathcal{M}_2 \star \Delta$ are equal.
 %\implies \mathbf{Y}_{\mathcal{M}_1}=\mathbf{Y}_{\mathcal{M}_2}. \]
 %\end{Lemma}
 %The result above clearly indicates that for it makes sense to consider formal input-output maps as opposed to 
 %trying to find conditions for the LFR to be well-posed. 
\end{Remark}
In fact, Remark \ref{rem:beck} suggests the following  intuitive interpretation:  \emph{the formal
input-output map determines the input-output behavior of an LFR for  bounded uncertainty and internally stabilizing control inputs.}

 %Indeed, if the LFR is stable, \emph{i.e.}, if it is well-posed, then 
 %formal input-to-state and input-output maps uniquely determine the intuitive input-to-state and input-output behavior of the LFR.
 In order to compare formally the behaviors of LFRs and ALPVs, we need to recall from \cite{Ball2005} some aspects of realization
 theory of LFRs. %In what follows, we will follow \cite{Ball2005}, with some modifications to suit our purposes. 
 %An LFR $\mathcal{M}$ is said to be a \emph{realization} of  a function
 %$f:\mathcal{X}^{*} \rightarrow \mathbb{R}^{p \times m}$, if $Y_{\mathcal{M}}=f$, \emph{i.e.}, $f$ equals the formal input-output map of $\mathcal{M}$. 
 We say that two LFRs $\mathcal{M}_1$ and $\mathcal{M}_2$ are \emph{formally input-output equivalent}, if their formal input-output maps are equals, \emph{i.e.},
 $Y_{\mathcal{M}_1}=Y_{\mathcal{M}_2}$.  If $\mathcal{M}$ is an LFR of the form \eqref{LFRform}, then we call the number $n=n_1+\cdots + n_d$ the \emph{dimension} of
 $\mathcal{M}$ and we denote it by $\dim \mathcal{M}$. We say that the LFR $\mathcal{M}$ is \emph{minimal}, if for any LFR $\mathcal{M}^{'}$
 which is formally input-output equivalent to $\mathcal{M}$, 
 %$f:\mathcal{X}^{*} \rightarrow \mathbb{R}^{p \times m}$, if $\mathcal{M}$ is a realization of $f$, and for any LFR $\mathcal{M}^{'}$ which is a realization of $f$, 
 $\dim \mathcal{M}  \le \dim \mathcal{M}^{'}$. 
% We call an LFR $\mathcal{M}$ \emph{minimal}, if it is a minimal realization of its own formal input-output map $Y_{\mathcal{M}}$. 
 
 Minimal LFRs can be characterized in terms of reachability and observability, and  minimal LFRs which are also input-output equivalent are in fact isomorphic. 
 In order to present this characterization formally, we need the following definitions. Let $\mathcal{M}$ be an LFR of the form \eqref{LFRform} and
 let $\{(H_i,F_{i,j},G_j)\}_{i,j=1}^d$ the corresponding canonical partitioning of $\mathcal{M}$. 
 Define the $k$-step observability $\{ \mathcal{O}_k^{i}(\mathcal{M})\}_{i=1}^{d}$ and $k$-step reachability matrices $\{ \mathcal{R}_k^{i}(\mathcal{M}) \}_{i=1}^{d}$ of $M$ recursively as follows:
 for all $i=1,\ldots,d$, %\red{READABILITY TO BE IMPROVED}
 \begin{equation*}
 \label{LFRreachobs}
 \begin{split}
  & \mathcal{O}_0^{i}(\mathcal{M}) = H_i, \quad  \mathcal{R}^i_0(\mathcal{M})=G_i, \\
  & \mathcal{R}_{k+1}^{i}(\mathcal{M}) \!\!\!=\!\!\!\begin{bmatrix} \mathcal{R}^i_0(\mathcal{M}),\!\! & F_{i,1}\mathcal{R}_k^1(\mathcal{M}), & \cdots,\!\!\! & F_{i,d}\mathcal{R}_k^d(\mathcal{M}) \end{bmatrix}, \!\!\! \\
  & \mathcal{O}_{k+1}^{i}(\mathcal{M}) \!\!\! =\!\!\! \begin{bmatrix} {\mathcal{O}_0^{i}(\mathcal{M})}^\top\!\!\!, &\!\!\!\! (\mathcal{O}_k^1(\mathcal{M}) F_{1,i})^\top\!\!,\!\!\! &\!\! \cdots, \!\!& \!\!\! (\mathcal{O}_k^d(\mathcal{M}) F_{d,i})^\top\end{bmatrix}^\top. \!\!\!
 \end{split}
 \end{equation*}
 We say that $\mathcal{M}$ is \emph{reachable} and \emph{observable}, if $\rank \mathcal{R}_{k}^i(\mathcal{M}) = n_i$ and $\rank \mathcal{O}_{k}^i(\mathcal{M}) =n_i$ for some $k>0$ respectively. %if $\max_{k \ge 0} \rank \mathcal{R}_{k}^i(\mathcal{M})=n_i$ for all $i=1,\ldots,d$. We say that  $\mathcal{M}$ is \emph{observable}, if $\max_{k \ge 0} \rank \mathcal{O}_k^i(\mathcal{M})=n_i$ for all $i=1,\ldots,d$. It can be shown \cite{Kaluzhnyy} that, in fact, $\max_{k \ge 0} \rank \mathcal{R}_{k}^i(\mathcal{M})=\rank \mathcal{R}_{n-1}^i(\mathcal{M})$ $\max_{k \ge 0} \rank \mathcal{O}_{k}^i(\mathcal{M})=\rank \mathcal{O}_{n-1}^i(\mathcal{M})$. 
 
 Hence, observability and reachability of LFRs can be verified numerically, and
 any LFR can be transformed to a reachable and observable LFR whose formal input-output map coincides with that of the original LFR. 
 Let $\mathcal{M}$ be an LFR of the form \eqref{LFRform} and let $\widetilde{\mathcal{M}}=(p,m,d,\tilde{A},\tilde{B},\tilde{C},\tilde{D})$. A nonsingular
 matrix $T \in \mathbb{R}^{n \times n}$ is said to be an \emph{isomorphism} from $\mathcal{M}$ to $\widetilde{\mathcal{M}}$, if 
 $D=\widetilde{D}$, $TAT^{-1}=\widetilde{A}$, $\widetilde{C}=CT^{-1}$, $\widetilde{B}=TB$ and $T=\diag{[T_1,T_2,\ldots,T_d]}$, where
 $T_i \in \mathbb{R}^{n_i \times n_i}$, $i=1,\ldots,d$.  Two LFRs are said to be isomorphic, if there exists an isomorphism from the one to the other.
 \begin{Theorem}[Minimality of LFRs, \cite{Ball2005}]
 \label{lfr:theo}
  An LFR is minimal if and only if it is reachable and observable. Two LFRs which are minimal and formally input-output equivalent are isomorphic.  
  Any LFR can be transformed to a minimal LFR which is formally input-output equivalent to the original one. 
 \end{Theorem}
 Note that stable LFRs in the sense of \cite{Beck1999} are closed under minimization. 
 \begin{Theorem}[\cite{Beck1999}]
  Any minimal LFR which is formally input-output equivalent to a $\gamma$-stable LFR is also $\gamma$-stable. 
 \end{Theorem}
 \begin{Remark}[Significance of minimality for control]
 \label{rem:min}
  If two LFRs are isomorphic, then they behave
%their corresponding LTI blocks are algebraically similar. Moreover, if $T$ is an isomorphism between two LFRs, then
% for any $\Delta \in \mathbf{\Delta}$, $T\Delta=\Delta T$. This means that isomorphic LFRs behave 
in the same manner when interconnected with a controller.
 Indeed, if the controller itself is an LFR, then its interconnection with two isomorphic LFRs yield two closed-loop systems which are also
 isomorphic LFRs.  In particular, if one of the closed-loop systems is stable (in the sense of Remark \ref{rem:beck}) then so is 
 the other, and vice versa, and the input-output behaviors defined by the star product of the two closed-loop systems are the same.
 Since all minimal and formally input-output equivalent LFRs are isomorphic, then any
 controller which stabilizes a minimal LFR and achieves certain input-output behavior  will  also stabilize
 and achieve the same input-output behavior for any other minimal and formally input-output equivalent LFR. 
 %Moreover, if the controller is an LFR itself, then the corresponding closed-loop
 %systems will be isomorphic and formaly input-output equivalent. In particular,  if the closed-loop 
 %system is a stable LFR in the sense of Remark \ref{rem:beck}, then the closed-loop system will be stable for any other minimal and formally
 %input-output equivalent LFR. In particular, all these closed-loop systems will be well-posed for any $\Delta \in \mathbf{\Delta}$ (see Remark \ref{rem:beck})
 %and their corresponding input-output maps, obtained by taking the star product with $\Delta$ will be equal. 
 To sum up, \emph{minimal and formally input-output equivalent LFRs yield the same closed-loop behavior when interconnected with any stabilizing LFR controller}. 
 This is why the preservation of minimality by ALPV to LFR transformation is so important. 
\end{Remark}

\subsection{Affine LPV Systems} \label{sec:ALPV}
Below, we recall some basic definitions for affine LPV models. We follow the terminology of \cite{PTM15}. 
A discrete-time Affine Linear Parameter-Varying (ALPV) model $(\Sigma)$ is defined as follows
\begin{equation}
\label{equ:FALPVSystem}
\Sigma\left\{
\begin{array}{lcl}
x(k+1) &=&  A(p(k)) x(k) + B(p(k)) u(k) , \\
y(k) &=& C(p(k)) x(k) + D(p(k)) u(k).
\end{array}
\right. 
\end{equation}
where $x(k) \in \mathbb{X}=\mathbb{R}^{n_\mathrm{x}}$ is the state vector, $y(k) \in \mathbb{Y}=\mathbb{R}^{n_\mathrm{y}}$ is the (measured) output signals, $u (k) \in \mathbb{U} = \mathbb{R}^{n_\mathrm{u}}$ represents the input signals while $p(k) \in \mathbb{P}= \mathbb{R}^{n_\mathrm{p}}$ is the scheduling variables of the system represented by $\Sigma$, and for all $p \in \mathbb{P}$,
%$A:\mathbb{P} \rightarrow \mathbb{R}^{n_{\mathrm x} \times n_{\mathrm x}}$, $B:\mathbb{P} \rightarrow \mathbb{R}^{n_{\mathrm x} \times n_{\mathrm u}}$, $C:\mathbb{P} \rightarrow \mathbb{R}^{n_{\mathrm y} \times n_{\mathrm x}}$ and $D:\mathbb{P} \rightarrow \mathbb{R}^{n_{\mathrm y} \times n_{\mathrm u}}$ defining the \ALPV representation \eqref{equ:FALPVSystem} are considered to be affine functions in the scheduling variable, \emph{i.e.}, there exist constant matrices
\begin{equation}\label{equ:FALPVSystemMatrices}
\begin{split}
& A(p) = A_0 + \sum_{i=1}^{{\NP}} A_i p_i, \quad  B(p) = B_0 + \sum_{i=1}^{{\NP}} B_i p_i,	 \\
& C(p) = C_0 + \sum_{i=1}^{{\NP}} C_i p_i, \quad  D(p) = D_0 + \sum_{i=1}^{{\NP}} D_i p_i,
\end{split}
%\begin{split}
%& A(p) = A_0 + \sum_{i=1}^{{\NP}} A_i \psi_i(p),	 \\
%& B(p) = B_0 + \sum_{i=1}^{{\NP}} B_i \psi_i(p),	 \\
%& C(p) = C_0 + \sum_{i=1}^{{\NP}} C_i \psi_i(p),	 \\
%& D(p) = D_0 + \sum_{i=1}^{{\NP}} D_i \psi_i(p). 
%\end{split}
\end{equation}
for constant matrices $A_i \in \mathbb{R}^{n_{\mathrm x} \times n_{\mathrm x}}, B_i  \in \mathbb{R}^{n_{\mathrm x} \times n_{\mathrm u}}, C_i \in \mathbb{R}^{n_{\mathrm y} \times n_{\mathrm x}}, D_i \in \mathbb{R}^{n_{\mathrm y} \times n_{\mathrm u}}$, $i \in \INP$.
In the sequel, we will use the short notation
\begin{equation*}
\Sigma=(\NP,\NX, \NU, \NY, \{ A_i, B_i, C_i, D_i\}_{i=0}^{{\NP}})
\end{equation*}
%to define an \ALPV of the form~\eqref{equ:FALPVSystem}.
to define a model of the form \eqref{equ:FALPVSystem}. The \emph{dimension of $\Sigma$} is the dimension $\NX$ of its state-space. Note that the system dimension $\NX$ does not depend on the number (dimension) of the scheduling parameters. 
%In the rest of the section, \emph{$\Sigma$ stands for an ALPV of the form \eqref{equ:FALPVSystem}}.
By a solution  of $\Sigma$ we mean a tuple of trajectories $(x,y,u,p)\in(\mathcal{X},\mathcal{Y},\mathcal{U},\mathcal{P})$ satisfying \eqref{equ:FALPVSystem} for all $k \in \mathbb{N}$, where 
$\mathcal{X}=\mathbb{X}^\mathbb{N}, \mathcal{Y}=\mathbb{Y}^{\mathbb{N}}, \mathcal{U}=\mathbb{U}^{\mathbb{N}},\mathcal{P}=\mathbb{P}^\mathbb{N}$, and we use the following notation:
%%In the sequel, we denote by $\mathbb{N}$ the set of natural numbers. 
for a set $\mathbb{A}$, we denote by $\mathbb{A}^{\mathbb{N}}$ the set of all functions of the form $\phi:\mathbb{N} \rightarrow \mathbb{A}$. An element of $\mathbb{A}^{\mathbb{N}}$ can be thought of as a signal in discrete-time.
%%%Finally, in the sequel we denote by $\mathbb{I}_0^{n_\mathrm{p}}$ the set $\{0,1,2,\ldots, n_\mathrm{p}\}$. 
%%\end{Notation}
%We say that $\Sigma$ is \emph{span-reachable}, if $\mathbb{X}$ is the linear span of all the vectors of the form
%$x_f=x(t)$, where $(x,y,u,p)$ is a solution of $\Sigma$, $x(0)=0$ and $t \in \mathbb{N}$. That is,
%$\Sigma$ is span-reachable, if the linear span of states reachable from the zero initial state yields the whole state-space.
%We say that $\Sigma$ is \emph{observable}, if for any $x_{1,0} \in \mathbb{X}$ and $x_{2,0} \in \mathbb{X}$, there exists an input
%$u \in \mathcal{U}$, a scheduling signal $p \in \mathcal{P}$, and two solutions $(x_1,y_1,u,p)$ and $(x_2,y_2,u,p)$ of
%$\Sigma$, such that $x_1(0)=x_{1,0}$, $x_2(0)=x_{2,0}$ and $y_1 \ne y_2$. 
%Next, we define the notion of input-output maps and input-to-state maps of $\Sigma$ induced by an initial state. 
%Let $x_0 \in \X$ be an initial state of $\Sigma$. 
Define the \emph{input-output function of $\Sigma$} as the function 
%\begin{subequations}
%\begin{align}
%X_{\Sigma,x_0} &:   \mathcal{U} \times \mathcal{P}  \rightarrow \mathcal{X}, \\
$Y_{\Sigma}:   \mathcal{U} \times \mathcal{P}  \rightarrow \mathcal{Y}$
%\end{align}
%\end{subequations}
such that for any $(x,y,u,p) \in \mathcal{X} \times \mathcal{Y} \times \mathcal{U} \times \mathcal{P}$,
%$x=X_{\Sigma,x_\mathrm{o}}(u,p)$ and 
$y=Y_{\Sigma}(u,p)$ holds if and only if
$(x,y,u,p)$ is a solution of $\Sigma$ and $x(0)=0$.
%The function $X_{\Sigma,x_{0}}$ is called the \emph{input-to-state function} of $\Sigma$ induced by the initial state $x_{0}$, and the function
%$Y_{\Sigma,x_{0}}$ is called the  \emph{input-to-output function} of $\Sigma$ induced by $x_{0}$. 
%We say that $\Sigma$ is \emph{span-reachable} from an initial state $x_0 \in \mathbb{X}$, if  $\mathrm{Span}\{ X_{\Sigma,x_0}(u,p)(t) \mid (u,p) \in \mathcal{U} \times \mathcal{P}, t \in \mathbb{N} \}=\mathbb{X}$. 
%We say that $\Sigma$ is \emph{observable}, if for any two states $x_1 \in \mathbb{X}$ and $x_2 \in \mathbb{X}$, $Y_{\Sigma,x_1} = Y_{\Sigma,x_2}$ implies $x_1 = x_2$. In the sequel, for the sake of simplicity, we will assume that the initial state of interest is zero.  We will say that $\Sigma$ is \emph{span-reachable}, if $\mathrm{Span}\{ X_{\Sigma,x_0}(u,p)(t) \mid (u,p) \in \mathcal{U} \times \mathcal{P}, t \in \mathbb{N} \}=\mathbb{X}$.
% We refer to the input-output map $Y_{\Sigma,0}$ induced by the initial state $x_0=0$ as the \emph{input-output map} of $\Sigma$, and we will denote $Y_{\Sigma,0}$ by $Y_{\Sigma}$. 
Two ALPVs are said to be \emph{input-output equivalent}, if their input-output maps coincide.  An \ALPV $\Sigma$ is said to be
\emph{a minimal}, if for any \ALPV\ $\hat{\Sigma}$ which is input-output equivalent to $\Sigma$, $\dim \Sigma \le \dim \hat{\Sigma}$. 

 From \cite{PM12,PTM15}, it follows that minimal ALPV systems can be characterized via observability and span-reachability, and input-output
 equivalent minimal ALPVs are isomorphic. In order to state this result precisely,
define
the \emph{$n$-step extended reachability matrix $\mathcal{R}_n(\Sigma)$ of $\Sigma$}, 
and \emph{$n$-step extended observability matrix $\mathcal{O}_n(\Sigma)$ of $\Sigma$}, $n \in \mathbb{N}$, recursively as follows:
%$n \in \mathbb{N}$, is defined recursively as follows
\begin{align*}
\mathcal{R}_0(\Sigma) &=
\begin{bmatrix}
  B_0,  & B_1, & \ldots, & B_{n_p}
\end{bmatrix}, \\
\mathcal{R}_{n+1}(\Sigma) &=
\begin{bmatrix}
  \mathcal{R}_0(\Sigma), & A_0 \mathcal{R}_n(\Sigma), & \ldots, & A_{n_\mathrm{p}} \mathcal{R}_n(\Sigma)
\end{bmatrix}, \\
%\end{align*}
%while the 
%\begin{align*}
\mathcal{O}_0(\Sigma) &= 
\begin{bmatrix}
C_0^\top, & C_1^{\top}, & \ldots,  & C_{n_\mathrm{p}}^\top
\end{bmatrix}^\top, \\
\mathcal{O}_{n+1}(\Sigma) &=
\begin{bmatrix}
\mathcal{O}_0^\top(\Sigma), & A_0^\top \mathcal{O}_{n}^\top(\Sigma),  & \ldots, & A_{n_\mathrm{p}}^\top \mathcal{O}^\top_n(\Sigma)
\end{bmatrix}^\top.
\end{align*}
%\end{Definition}
Let us call $\Sigma$ \emph{span-reachable}, if $\Rank \mathcal{R}_{\NX-1}=\NX$, and let us call $\Sigma$ \emph{observable}, if
$\Rank \mathcal{O}_{\NX-1}=\NX$. 
%Note that an equivalent defintion of span-reachability and observability of an ALPV can be formulated
%in terms of properties of its trajectoris, see \cite{PM12,PTM15}. 
Finally, we % recall the notion of isomorphism for \ALPV models.
%\begin{Definition}[\textbf{\ALPV isomorphism}]\label{sect:problem_form:lin:morphism}
consider an ALPV $\Sigma$ of the form \eqref{equ:FALPVSystem}, and an ALPV
 $\Sigma'=(\NP,\NX, \NU, \NY,\{A_i', B_i', C_i', D_i'\}_{i=0}^{\NP})$ with $dim(\Sigma)=dim(\Sigma')=\NX$. A nonsingular matrix $\MORPH \in \mathbb{R}^{\NX \times \NX}$ is said to be an \emph{\ALPV\  isomorphism} from $\Sigma$ to $\Sigma'$, if for all $i=0,\ldots,\NP$
\begin{align*}
 A_i'\MORPH &= \MORPH A_i,  & B_i'&= \MORPH B_i,  & C_i'\MORPH&= C_i, & D_i'&= D_i .
\end{align*}
Now we can recall the following result from \cite{PM12}.
\begin{Theorem}[ \cite{PM12,PTM15}] %An \ALPV $\Sigma$ is
\label{min:theo}
%\begin{itemize}
%\item[-]
 %span-reachable if and only if its $\NX-1$th step extended controllability matrix $\mathcal{R}_{\NX-1}(\Sigma)$ is of rank $\NX$, and $\Sigma$ is observable if and only if its $\NX-1$th step extended observability matrix $\mathcal{O}_{\NX-1}(\Sigma)$ is of rank $\NX$.
%\item[-]
 An ALPV is minimal if and only if it is span-reachable and observable.  
 Any two minimal ALPVs which are input-output equivalent are isomorphic.  Any ALPV can be transformed to an input-output equivalent minimal ALPV.
\end{Theorem}

\subsection{Transforming ALPVs to LFRs}\label{sec:LPV_LFR_Transformation}
 One popular approach for control of ALPVs is to transform them to LFRs as follows \cite{Verdult2002}.

 \begin{Definition} 
\label{LPV2LFR}
 Let $\Sigma$ be an ALPV of the form \eqref{equ:FALPVSystem}. An LFR-LPV 
$\mathcal{M}$ of the form \eqref{LFRform} is called an \emph{LFR calculated from
 $\Sigma$}, if $d=\QNUM+1$, $n_1=\NX$, $p=\NY$, $m=\NU$, $D=D_0$ and the canonical decomposition 
 $\{(H_i, F_{i,j},G_j)\}_{i,j=1}^{d}$ of $\mathcal{M}$ satisfies the following properties:
 \begin{enumerate}
 \item $H_1=C_0$, $G_1=B_0$, $F_{1,1}=A_0$,
 \item for all $i,j=1,\ldots, d$,  if $i > 1$ and $j > 1$, then $F_{i,j}=0$,
 \item for all $i=2,\ldots,d$,
%$\mathcal{M}$
% can be partitioned as 
% \[
%         \begin{split}
%          & A=\begin{bmatrix} A_0 & B_{w} \\ C_z & 0 \end{bmatrix}, ~ B=\begin{bmatrix} B_0 \\ D_{zu} \end{bmatrix}, C=\begin{bmatrix} C_0 & D_{yw} \end{bmatrix} \\
%           & B_w = \begin{bmatrix} B_{w_1},& \ldots, & B_{w_d}  \end{bmatrix}, \quad D_{yw} = \begin{bmatrix} D_{yw_1},& \ldots, & D_{yw_d}  \end{bmatrix}, \\
%           & C_z=\begin{bmatrix} C_{z_1}^\top & \ldots, & C_{z_d}^\top \end{bmatrix}^\top, \quad 
%          D_{zu}= \begin{bmatrix} D_{zu_1}^\top & \ldots, & D_{zu_d}^\top \end{bmatrix}^\top, \\
%         \end{split}
%  \]
%  where $C_{z_i} \in \mathbb{R}^{n_{i+1} \times n_1}$,  $D_{zu_i} \in \mathbb{R}^{n_{i+1} \times m}$,
%  $B_{w_i} \in \mathbb{R}^{n_1 \times n_{i+1}}$, 
%  $D_{yw_i} \in \mathbb{R}^{p \times n_{i+1}}$, $i=1,\ldots, d-1$ and 
  \begin{equation}
			\label{LFRpatrialMatrices}
	\begin{bmatrix}
	A_{i-1} & B_{i-1} \\
	C_{i-1} & D_{i-1}  \\
	\end{bmatrix} = 
	\begin{bmatrix}
	F_{1,i} \\
	H_{i}  \\
	\end{bmatrix}
	\begin{bmatrix}
	F_{i,1} & G_i \\
	\end{bmatrix}.
	\end{equation}
      %and  the matrices   $\begin{bmatrix}C_{z_i} & D_{zu_i}\end{bmatrix}$ and $\begin{bmatrix}
      %		B_{w_i}^\top&C_{z_i}^\top
      %		\end{bmatrix}^\top$ are full row and column rank respectively, for all $i=1,\ldots,d$.
\end{enumerate}
\end{Definition}
 The intuition behind Definition \ref{LPV2LFR} is as follows. A solution $(x,y,u,p)$, $x(0)=0$ of the ALPV model $\Sigma$ corresponds 
 to a solution of the LFR $\mathcal{M}$ for the following  choice of $\Delta$
\begin{equation}
\label{DeltaScheduling}
 \Delta=\Delta(p)= \diag{[\lambda I_{n_1}, \delta_1 I_{n_2},\ldots,\delta_d I_{n_{d}}]},
\end{equation}
 with $\delta_i,\lambda: (\mathbb{R})^{\mathbb{N}} \rightarrow (\mathbb{R})^{\mathbb{N}}$ defined by
 $\delta_i(h)(t)=p_i(t)h(t)$ and $\lambda(h)(t)=\left\{\begin{array}{rl} h(t-1) & t > 0 \\ 0 & t=0 \end{array}\right.$ for any sequence $h \in (\mathbb{R}^{\mathbb{N}})$ and $t \in \mathbb{N}$, \emph{i.e.},
 \begin{equation*}
\label{Mmatrix1}
\begin{split}
& \begin{bmatrix} \Delta(p) \begin{bmatrix} x \\ z \end{bmatrix} \\ y \end{bmatrix} =\begin{bmatrix}
A_0 & B_w & B_0\\
C_{z} & 0 & D_{zu} \\
C_0 & D_{yw} & D_0 \\
\end{bmatrix} \begin{bmatrix} x  \\ w   \\  u \end{bmatrix},  \\
& B_w=\begin{bmatrix} F_{1,2}, & \ldots, & F_{1,d} \end{bmatrix}, ~ C_z=\begin{bmatrix} F_{2,1}^\top,& \ldots,& F_{d,1}^\top \end{bmatrix}^\top, \\
& D_{zu}=\begin{bmatrix} G_{2}^\top,& \ldots,& G_{d}^\top \end{bmatrix}^\top, ~ D_{yw}=\begin{bmatrix} H_{2},& \ldots,& H_{d} \end{bmatrix}. \\
%& \begin{bmatrix} v & w =\Deltaz(t), \quad 
\end{split}
\end{equation*} 
 The specific form of LFRs calculated from ALPVs serves as a motivation to define the following subset of LFRs.
 \begin{Definition}[LPV-LFR]
  An LPV-LFR is an LFR of the form \eqref{LFRform}, 
  %\begin{equation}
  % \label{LPVLFRform}
  % \mathcal{M}=(p,m, d+1, \{n_i\}_{i=1}^{d+1},A,B,C,D), 
  %\end{equation}
    where $d > 1$, and the canonical decomposition $\{(H_i, F_{i,j},G_j)\}_{i,j=1}^{d}$ of $\mathcal{M}$ has the property that
    $F_{i,j}=0$, if $i > 1$ and $j > 1$. 
  %the matrices $A,B,C$ admit a decomposition
  %\begin{equation}
  %\label{Mmatrix}
  % \begin{split}
  % & A=\begin{bmatrix} A_0 & B_{w} \\ C_z & 0 \end{bmatrix}, ~ B=\begin{bmatrix} B_0 \\ D_{zu} \end{bmatrix}, C=\begin{bmatrix} C_0 & D_{yw} \end{bmatrix}.
 % \%end{split}
 % \%end{equation}
 %where $A_0 \in \mathbb{R}^{n_1 \times n_1}$, $B_w \in \mathbb{R}^{n_1 \times n_z}$,  $C_z \in \mathbb{R}^{n_z \times n_1}$, $B_0 \in \mathbb{R}^{n_1 \times m}$, 
 %$D_{zu} \in \mathbb{R}^{n_z \times m}$, $C_0 \in \mathbb{R}^{p \times n_1}$, $D_{yw} \in \mathbb{R}^{p \times n_z}$,
 %$n_z=n_2+\cdots+n_{d}$.  
\end{Definition}
 It then follows that an LFR calculated from an ALPV model is an LPV-LFR.
 Note that not only ALPV models can be transformed to LPV-LFR models, but there is a transformation in the reverse direction.  
  \begin{Definition}[From LPV-LFR to ALPV]
 \label{LFR2LPV}
  Let $\mathcal{M}$ be an LPV-LFR  of the form \eqref{LFRform} and let $\{(H_i, F_{i,j},G_j)\}_{i,j=1}^{d}$ be its canonical
  decomposition. We can define the \emph{ALPV $\Sigma_{\mathcal{M}}$ which corresponds to $\mathcal{M}$} as 
  the ALPV of the form \eqref{equ:FALPVSystem}, such that $\NP=d-1$, $\NX=n_1$, $\NY=p,\NU=m$, $D=D_0$, 
  $A_0=F_{1,1},B_{0}=G_1$, $C_0=H_1$, and for all $i=1,\ldots,\NP$, 
 \[ 	\begin{bmatrix}
	A_i & B_i \\
	C_i & D_i  \\
	\end{bmatrix} = 
	\begin{bmatrix}
	F_{i+1,1} \\
	H_{i+1,1}  \\
	\end{bmatrix}
	\begin{bmatrix}
	F_{1,i+1} & G_{i+1} \\
 	\end{bmatrix}.
  \]
   %where the matrices $C_{z_i} \in \mathbb{R}^{n_{i+1} \times n_1}$,  $D_{zu_i} \in \mathbb{R}^{n_{i+1} \times m}$,
  %$B_{w_i} \in \mathbb{R}^{n_1 \times n_{i+1}}$, 
  %$D_{yw_i} \in \mathbb{R}^{p \times n_{i+1}}$ are such that
   %\[
   %\begin{split}
   %& A=\begin{bmatrix} A_0 & B_{w} \\ C_z & 0 \end{bmatrix}, ~ B=\begin{bmatrix} B_0 \\ D_{zu} \end{bmatrix}, C=\begin{bmatrix} C_0 & D_{yw} \end{bmatrix} \\
    %       & B_w=\begin{bmatrix} B_{w_1},& \ldots, & B_{w_d}  \end{bmatrix}, ~ C_z = \begin{bmatrix} C_{z_1}^\top & \ldots, & C_{z_d}^\top \end{bmatrix}^\top, \\
    %       & D_{zu} =\begin{bmatrix} D_{zu_1}^\top & \ldots, & D_{zu_d}^\top \end{bmatrix}^\top, ~D_{yw}= \begin{bmatrix} D_{yw_1},& \ldots, & D_{yw_d} \end{bmatrix}
    %\end{split}
  %\]
 \end{Definition}
 Note that, while there are many ways to transform an ALPV to and LPV-LFR, each LPV-LFR gives rise to a single ALPV.
 The operation of transforming an LPV-LFR to an ALPV is in a sense the inverse of the transformation of an ALPV to LPV-LFR:
 if $\mathcal{M}$ is an LPV-LFR calculated from an ALPV $\Sigma$ using Definition \ref{LPV2LFR}, 
 then $\Sigma_{\mathcal{M}}=\Sigma$. However,
 $\mathcal{M}$ is any LPV-LFR, then the LPV-LFRs calculated from $\mathcal{M}_{\Sigma}$ using Definition \ref{LPV2LFR} 
are in general different from $\mathcal{M}$ and they need not even be isomorphic to $\mathcal{M}$.  

  That is, ALPVs yield LPV-LFRs and LPV-LFRs can be converted to ALPVs. Intuitively, the conversion is such that one could use
  control design techniques for LFRs to control ALPVs. If this path is taken, then the sole use of ALPV models is to serve
  as a source of LFR models, and hence instead of identifying ALPV models, one could identify LFR models directly.
  Unfortunately, as we will see in the next section, without further restrictions on the transformation from ALPV to LFR,
  this may lead to inconsistent results.

  \subsection{Inconsistency of the ALPV to LFR transformation: motivating example}
   \label{problem:form} 
   The transformations of Definition \ref{LPV2LFR} and Definition \ref{LFR2LPV} give rise to a number of fundamental
   questions, which have implications for control and system identification.  
   %To begin with, the relationship between
   %the input-output functions of ALPVs and the formal input-output map of the LFRs computed from them is not clear.
 % \textbf{ Can it happen that two ALPVs which are input-output equivalent yield LFRs which are not formaly  input-output equivalent ?}
 %  If the answer is affirmative, then we cannot expect that the same controller will work for both LFRs, especially if one
 %  uses methods from robust control, which are supposed to work for any choice of the uncertainity block $\Delta$. 
 %  An affirmative answer can lead to problems for identification of LFRs, if the measured data correspond to 
 %  solutions of the LFR which correspond to $\Delta=\Delta(p)$ from \eqref{DeltaScheduling}.
 %  Note that the transformation in Definition \ref{LPV2LFR} is not unique, so the question cannot be dimissed even for
 %  the case of one single ALPV.
   To illustrate these problems, let us consider the following example.
   \begin{Example}[Motivating example]
   \label{Example1}
    Let us consider the \ALPV model $\Sigma$ of the form \eqref{equ:FALPVSystem}, such that $\NP = 1, \NX=2, \NU = \NY =1$, with the following model matrices
%We take example~1 from \cite{Alkhoury15}, where it is proven that this system is minimal.
\begin{align*}
A_0 = \begin{bmatrix}
1& 0\\ 0 &0.2
\end{bmatrix}, B_0=\begin{bmatrix} 1\\0\end{bmatrix}, C_0^\top = \begin{bmatrix}
1 \\0
\end{bmatrix},\\
A_1 = \begin{bmatrix}
0&2\\1&1
\end{bmatrix}, B_1=\begin{bmatrix}0\\1\end{bmatrix}, C_1^\top = \begin{bmatrix}
0 \\1
\end{bmatrix},
\end{align*}
and $D_0 = D_1 = 0$.  Consider now the following two LFRs,
\[ 
 \begin{split}
 & \mathcal{M}=(1,1,2,\{2,3\}, (A,B,C,0)), \\
 & \mathcal{\tilde{M}}=(1,1,2,\{2,3\}, (\tilde{A},\tilde{B},\tilde{C},0)) , \\
%A=\begin{bmatrix} 1 & 0 &  -1.196  &  0.5429 \\
%                  0 & 0.2 &  -0.8668 &-0.9364 \\
%                   -0.3413  & -1.519 & 0 & 0 \\
%                 -0.752   &  0.338   & 0 & 0 
%\end{bmatrix}
%& B=\begin{bmatrix} 1 & 0 & -0.3413 &  -0.752 \end{bmatrix}^T \\
%& C=\begin{bmatrix}  
%1  &  0    &    -0.598  & 0.2714 \end{bmatrix} \\
& A= \begin{bmatrix} 1 & 0 & 1   &  0  &    1 \\
                               0 & 0.2 & 1 &    0  &   0  \\
                               1 &  1 & 0 & 0 & 0 \\
                               0 &  2 & 0 & 0 & 0 \\
                               -1 &  1 & 0 & 0 & 0 
                \end{bmatrix}, ~ B=\begin{bmatrix} 1 \\ 0 \\ 1  \\ 200  \\ -1 \end{bmatrix},\\
& C=\begin{bmatrix} 1 & 0 & 0.5 & 0 & 0.5 \end{bmatrix}, \\
& \tilde{A}= \begin{bmatrix} 1 & 0 & 0   &  1  &    0 \\
                               0 & 0.2 & 1 &    0  &   0  \\
                               1 &  1 & 0 & 0 & 0 \\
                               0 &  2 & 0 & 0 & 0 \\
                               0 &  -2 & 0 & 0 & 0 
                \end{bmatrix}, ~ \tilde{B}=\begin{bmatrix} 1 \\ 0 \\ 1 \\ 0  \\  1 \end{bmatrix}, \\
& \tilde{C}=\begin{bmatrix} 1 & 0 & 0 & 0.5 & 0 \end{bmatrix}.
\end{split}
\]
It is easy to see that both $\mathcal{M}$ and $\mathcal{\tilde{M}}$ satisfy Definition \ref{LPV2LFR}, yet the matrices
are completely different. In fact, it is easy to see that $\mathcal{M}$ and $\mathcal{\tilde{M}}$ are not
isomorphic. At a first glance, it is not clear that these two LFRs are in fact formally input-output equivalent. 
Concerning system identification, this example raises the question as \textbf{to how to distinguish between these two
LFRs, since clearly they originate from the same ALPV, and hence their behavior for $\Delta=\Delta(p)$ from \eqref{DeltaScheduling}
should be the same}.  
From the point of view of control, these two LFRs behave quite differently. Using classical $H_{\infty}$ control, 
we computed LTI controllers for both LFRs which render the closed-loop $\gamma$-stable. If this controller is applied to
the original ALPV, then it renders it stable for all scheduling sequences $p \in \mathcal{P}$ satisfying
$|p(t)| < \gamma$ for all $t \in \mathbb{N}$. The largest $\gamma$ we could get with $\mathcal{M}$ is $\frac{1}{4}$, while
the largest $\gamma$ we could get for $\mathcal{\tilde{M}}$ is $\frac{1}{281}$.
That is, the \textbf{guaranteed performance of the controller depends on the choice of the LFR!}

The situation becomes even more interesting, if we notice that the following LFR
$\mathcal{\hat{M}}=(1,1,2,\{2,2\},\hat{A},\hat{B},\hat{C},0)$ with 
\[ 
 \begin{split}
& \hat{A}=\begin{bmatrix} 1 & 0 &  -1.196  &  0.5429 \\
                  0 & 0.2 &  -0.8668 &-0.9364 \\
                   -0.3413  & -1.519 & 0 & 0 \\
                 -0.752   &  0.338   & 0 & 0 
\end{bmatrix}, \\
&  \hat{B}=\begin{bmatrix} 1 & 0 & -0.3413 &  -0.752 \end{bmatrix}^T, \\
& \hat{C}=\begin{bmatrix}  1  &  0    &    -0.598  & 0.2714 \end{bmatrix}.
\end{split}
\]
also satisfies Definition \ref{LPV2LFR}. The dimension of $\mathcal{\hat{M}}$ is smaller than the dimension of
$\mathcal{M}$ and $\mathcal{\tilde{M}}$. This means
$\mathcal{M}$ and $\mathcal{\tilde{M}}$ are not minimal dimensional LFR representations of the ALPV $\Sigma$, and hence
we might be tempted to think that our problems are caused by parasitic dynamics which are present in 
$\mathcal{M}$ and $\mathcal{\tilde{M}}$, but which are absent from the ALPV $\Sigma$. 
But how can we be sure $\mathcal{\hat{M}}$ is itself minimal? How to modify Definition \ref{LPV2LFR}, so that
we cannot get LFRs of higher dimension than it is strictly necessary?
%What is the minimal dimensional LFR representation of an ALPV ?
What if we can find another minimal dimensional LFR which satisfies Definition \ref{LPV2LFR} and which is not
isomorphic to $\mathcal{\hat{M}}$?  
\end{Example}

\subsection{Problem formulation}\label{sec:problem_formulation}
The questions raised above can be addressed by answering the following questions:
\begin{enumerate}
\item Is it true that two ALPVs which are input-output equivalent yield LFRs which are formally  input-output equivalent?
  %If the answer is negative, then we cannot expect that the same controller will work for both LFRs, especially if one
  %uses methods from robust control, which are supposed to work for any choice of the uncertainity block $\Delta$. 
  % An affirmative answer can lead to problems for identification of LFRs, if the measured data correspond to 
  % solutions of the LFR which correspond to $\Delta=\Delta(p)$ from \eqref{DeltaScheduling}.
\item Can we modify Definition \ref{LPV2LFR} so that minimal ALPVs get transformed to minimal LFRs in the sense of
           Theorem \ref{lfr:theo}, and that this transformation preserves isomorphism? 

\item Is the ALPV calculated from a minimal LPV-LFR according to Definition \ref{LFR2LPV} minimal?
\item Can we transform an LPV-LFR to a minimal LFR which is also an LPV-LFR?

\end{enumerate}

 If the answers to these questions are positive, then the situation in Example \ref{Example1} can be handled easily,
 by using Theorem \ref{lfr:theo} and Theorem \ref{min:theo}.
 Namely, it is enough to restrict attention to minimal ALPVs and LFRs. Then, the modified transformation will guarantee that minimal and input-output equivalent ALPVs are transformed to minimal and formally input-output equivalent, and thus isomorphic LFRs.
 That is, the result of transforming ALPVs to LFRs is essentially unique, when minimal ALPVs are concerned. 
 From Remark \ref{rem:min}, it then follows that the result of control synthesis will not depend on which particular
 minimal ALPV or LFR was chosen, as long as the chosen ALPV is input-output equivalent to the original one. 
 If one would like to bypass ALPVs altogether, then one should work with minimal LPV-LFRs, where minimality is understood in the sense of Theorem \ref{lfr:theo}.   Then, if two minimal LPV-LFRs describe two input-output equivalent ALPVs, then they will be
 isomorphic, and hence equivalent for control synthesis. Moreover, with some more work we can show that identifiability
of LPV-LFRs is equivalent to that of ALPVs, so for identification it will not matter which modelling framework is used. 

In the next section we state the answers to the questions above formally.

\section{Equivalence between ALPV and LFR-LPV: preservation of input-output behavior, minimality and identifiability}
\label{sec:Connections_minimality_LPV_LFPV}

We start with presenting a special case of Definition \ref{LPV2LFR}, which will have some useful properties.
\begin{Definition}[MR factorization]
Let $\Sigma$ be an ALPV of the form \eqref{equ:FALPVSystem}.
We say that an LFR $\mathcal{M}$ of the form \eqref{LFRform} is \emph{calculated by Matrix Full Rank (MR) factorization from $\Sigma$},
if $\mathcal{M}$ satisfies Definition \ref{LPV2LFR}, and in addition, for any $i=2,\ldots,d$,
$\begin{bmatrix} F_{i,1} & G_{i}\end{bmatrix}$ and $\begin{bmatrix}
      		F_{1,i}^\top& H_{i}^\top \end{bmatrix}^\top$ are full row and column rank respectively, where
$\{(H_i,F_{i,j},G_j)\}_{i,j=1}^d$ is the canonical partitioning of $\mathcal{M}$.
\end{Definition}
That is, the only distinguishing feature of MR factorization is that it explicitly requires the
 factorization \eqref{LFRpatrialMatrices} to be full rank. 
 %\begin{Remark}%[Uniqueness of LPV-LFR calculated by MR factorization]
  Note that LPV-LFRs calculated by MR factorization are unique up to isomorphism. 
%Indeed, if the matrices $B_{w},D_{yw},C_{z}, D_{zu}$ satisfy \eqref{LFRpartialMatrices}, then so do
%  $B_wS,D_{yw}S,S^{-1}C_z,SD_{zu}$, for any matrix $S=\diag{[S_1,\ldots,S_d]}$, where $S_i \in \mathbb{R}^{n_{i+1} \times n_{i+1}}$ and $S_i$ is non-singular for all
 % $i=1,\ldots,d$, and all other solutions of \eqref{LFRpartialMatrices} can be obtained from $B_w,D_{yw},C_z,D_{zu}$ in this way. 
 %\end{Remark}
 %It is easy to see that the transformation in Definition \ref{LPV2LFR} preserves input-output behavior:
 %\begin{Lemma}
 %\label{theo1}
 %  Let $\Sigma$ be an ALPV and let $\mathcal{M}$ be an LPV-LFR calculated from $\Sigma$ by a MR factorization. %from Definition \ref{LPV2LFR}. 
   %For any $u \in \mathcal{U}, p \in \mathcal{P}$, $(x,p,u,y)$ is a solution of $\Sigma$ if and only if
   %$(x,z,p,u,y)$ is a solution of $\mathcal{M}_{\Sigma}$ with $z=C_zx+D_{zu}u$. 
 %   It is easy to see that the input-output maps of
 %  $\Sigma$ and $\mathbb{Y}_{\mathcal{M}_{\Sigma}}$ coincide. 
  %$Y_{\Sigma}=\mathbf{Y}_{\Sigma}$.
 %\end{Lemma}
% Theorem \ref{theo1} is a simple, and well-known result. 
%However, this result does not tell us what happens if the uncertainty block $\Delta(p(t))$ of $\mathcal{M}_{\Sigma}$
% is replaced by uncertainty block which contains uncertainties others than the one induced by scheduling parameters. Moreover, it remains unclear if properties like isomorphism, minimality, etc., are preserved when passing from ALPVs to LFRs.

 The next theorem, which is the main result of the paper, tells us that MR transformation preserves minimality and isomorphism, and
 it maps input-output equivalent ALPVs to formally input-output equivalent LFRs. 

\begin{Theorem}[Transforming ALPV to LPV-LFRs]\label{th:LFPVMinimality}
       Let $\Sigma,\Sigma_1,\Sigma_2$ be ALPVs and let $\mathcal{M},\mathcal{M}_1,\mathcal{M}_2$ be an LPV-LFR calculated from $\Sigma,\Sigma_1,\Sigma_2$ respectively by MR factorization.
    \begin{enumerate}
    \item
       $\Sigma$ is a minimal ALPV $\iff$ $\mathcal{M}$ is a  minimal LFR.
    \item 
	The \ALPVs $\Sigma_1$ and $\Sigma_2$ are input-output equivalent $\iff$
        $\mathcal{M}_{1}$ and $\mathcal{M}_{2}$ are formally input-output equivalent.
   \item 
      $\Sigma_1$ and $\Sigma_2$ are isomorphic $\iff$ the corresponding LFRs $\mathcal{M}_{1}$ and $\mathcal{M}_{2}$ are isomorphic. 
   \end{enumerate}
\end{Theorem}
 %That is, the transformation from ALPVs to LFRs preserves isomorphism, minimality and in fact represents a one-to-one map between input-output behaviors of
 %ALPVs and LFRs.  

Note that Theorem \ref{th:LFPVMinimality} allows us to derive the following useful properties for the transformation from
LPV-LFRs to ALPVs.
%The situation changes when we restrict attention to minimal LFRs.
 \begin{Theorem}[Transformation from LPV-LFR to ALPV]
 \label{LFR2LPV:theo}
 Let $\mathcal{M},\mathcal{\tilde{M}},\mathcal{\hat{M}}$ be LPV-LFRs, and let $\Sigma=\Sigma_{\mathcal{M}}$, $\tilde{\Sigma}=\Sigma_{\mathcal{M}}$  
be the ALPVs associated with $\mathcal{M}$ and $\mathcal{\tilde{M}}$ respectively. %Then the following holds.
 \begin{itemize}
 \item[{\textbf{(i)}}] If $\mathcal{M}$ is minimal, then $\Sigma$ is a minimal ALPV.
                   
 \item[{\textbf{(ii)}}] $\mathcal{M}$ and $\mathcal{\tilde{M}}$ are formally input-output equivalent, if and only if 
                       $\Sigma$ and $\tilde{\Sigma}$ are input-output equivalent.

 \item[{\textbf{(ii)}}] If $\mathcal{M}$ and $\mathcal{\tilde{M}}$ are isomorphic, then 
                       $\Sigma$ and $\tilde{\Sigma}$ are isomorphic.

 \item[{\textbf{(iv)}}] Let $\mathcal{\hat{M}}$ be the LPV-LFR  computed from $\Sigma$ by MR factorization. Then
                      $\mathcal{M}$ and $\mathcal{\hat{M}}$ are \emph{formally} input-output equivalent. 
                     If $\mathcal{M}$ is minimal, then so is $\mathcal{\hat{M}}$ and it is isomorphic to $\mathcal{M}$.
 \end{itemize}
 \end{Theorem}
% Note that $\mathcal{M}$ and $\mathcal{\hat{M}}$  from Theorem \ref{LFR2LPV:theo} are trivially input-output equivalent. Part \textbf{(i)} of Theorem \ref{LFR2LPV:theo} 
% says that in fact they are formaly input-output equivalent, i.e. they will yield the same input-output behavior for uncertainity blocks which do not necessarily
%th:LFPVMinimality arise frol a scheduling signal. 
%Part \textbf{(i)} says that  the transformation from Definition \ref{LFR2LPV} and MR factorization are each others' inverses. 
Parts \textbf{(i)} -- \textbf{(iii)} of Theorem \ref{LFR2LPV:theo} say that the transformation from LPV-LFR to ALPVs preserves
 input-output equivalence,  minimality and isomorphism. Finally, Part \textbf{(iv)} says that,
if attention is restricted to minimal LFRs, then the 
 transformations from Definition \ref{LFR2LPV} and Definition \ref{LPV2LFR} are each others' inverses, if isomorphic models are viewed as equals.

 Theorem~\ref{LFR2LPV:theo} and Theorem~\ref{th:LFPVMinimality} have several consequences.
 \begin{Corollary}
 \label{col2}
   Any minimal LFR which is formally input-output equivalent to an LPV-LFR is also an LPV-LFR. 
 \end{Corollary}
  In other words, Corollary \ref{col2} states that the class of LPV-LFRs are closed under minimization, \emph{i.e.}, when working with
  LPV-LFRs we can always assumed they are minimal. 
  
  Let us say that two LPV-LFRs are \emph{input-output equivalent}, if their associated ALPVs, as defined in Definition \ref{LFR2LPV}, are
  input-output equivalent. Intuitively, if two LPV-LFRs are input-output equivalent, then their input-output behavior is the same
  for any $\Delta$ of the form \eqref{DeltaScheduling}. From Theorem \ref{LFR2LPV:theo} and Theorem \ref{th:LFPVMinimality}
  we can deduce the following.
  \begin{Corollary}
 \label{col1}
   Two LPV-LFRs are input-output equivalent, if and only if they are formally input-output equivalent.
 \end{Corollary}
 That is, for LPV-LFRs, their formal input-output map uniquely determines the input-output function of the associated ALPV!
 In fact,
 the proof of Theorem \ref{LFR2LPV:theo} and Theorem \ref{th:LFPVMinimality} allows us to characterize the formal 
 input-output maps of LPV-LFRs.
 \begin{Theorem}\label{th:LFR_LPV_equivalence}
	An LFR $\mathcal{M}$ is formally input-output equivalent to an LPV-LFR, if and only if
	$Y_{\mathcal{M}}(s)=0$ for all $s \in \mathcal{X}^{*}$, such that 
	$s$ not of the form in $i_11i_2 \cdots 1 i_k1$ for some $i_1,\ldots,i_k \in \{1,\ldots,d\}$, $k > 0$. 

        %If $\mathcal{M}$ is minimal, and the condition above holds, then $\mathcal{M}$ is an LPV-LFR. 
\end{Theorem}
 This allows us to determine if an arbitrary LFR  can be represented as a LPV-LFR. 

%\section{Application: identifiability of LPV-LFRs}
The results  of Theorem \ref{th:LFPVMinimality} allow us to conclude that identifiability of
ALPV and LFRs parameterizations is in fact equivalent.
% provided that one restricts attention to minimal ALPV and LFR models. 
%In turn, this equivalence and the existing results for identifiability of ALPV models from \cite{Alkhoury15} allow us to characterize identifiability of LPV-LFR models. 
%Note that if one uses non-minimal models, then it is possible to come up with a parametrizion of LPV-LFR models and ALPV models which describe the same
%set of input-output behaviors, but one of which is identifiable, and the other one is not. % transformation from Definition \ref{LPV2LFR}. 
%Finally, we present the applications of the results above to identifiability of LPV-LFRs. 
In order to present the results formally, we define the concepts of LPV-LFR parameterizations and their structural identifiability.
To this end, let us fix integers $d, m$ and $p$, and 
denote by $\mathcal{LFR}(d,m,p,\{n_i\}_{i=1}^{d})$ the set of all LPV-LFRs of the form \eqref{LFRform}.
% That is, $\mathcal{LFR}(d,m,p,\{n_i\}_{i=1}^{d})$ is the
%set of all LPV-LFRs with $m$ inputs, $p$ outputs, $d$ scheduling variables, $n$ states and with the uncertainity blocks $\Delta(p)$ being of the form
%$\Delta(p)=diag{[p_1I_{n_2}, \ldots, p_dI_{n_{d+1}}]}$. 
%$\mathcal{LPV}(\NP,\NX,\NU,\NY)$ is the set of all ALPV modes with $\NX$ states, $\NY$
%outputs, $\NU$ inputs and $\NP$ scheduling parameters. 
\begin{Definition}[Parameterizations]
Let $\Theta \subseteq \mathbb{R}^{\NTH}$ be the space of parameters.
An LPV-LFR parametrization is a function $\mathbf{M}: \Theta \rightarrow \mathcal{LFR}(d,m,p,\{n_i\}_{i=1}^{d})$. 
\end{Definition}
We will say that an LPV-LFR parametrization $\mathbf{M}$ is \emph{structurally identifiable}, if for any two distinct parameter values $\theta_1, \theta_2 \in \Theta$, $\theta_1 \ne \theta_2$, $\mathbf{M}(\theta_1)$ and $\mathbf{M}(\theta_2)$ are not input-output equivalent. We say that the LPV-LFR parametrization $\mathbf{M}$ is
\emph{formally structurally identifiable}, if for any two distinct parameter values $\theta_1, \theta_2 \in \Theta$, $\theta_1 \ne \theta_2$, $\mathbf{M}(\theta_1)$ and $\mathbf{M}(\theta_2)$ are not formally input-output equivalent. 
%Finally, $\mathbf{M}$ is called \emph{structurally minimal}, if for every parameter value $\theta \in \Theta$,
%$\mathbf{M}(\theta)$ is a minimal LFR. Note that difference between structural identifiability and formal structural identifiability of LPV-LFR parametrizations. 
Structural identifiability means that the parameter values can be uniquely determined by observing the output of the underlying system for some input and for some choice of the
uncertainity block $\Delta$ of the form \eqref{DeltaScheduling}, 
which corresponds to a choice of the scheduling variables. In contrast, formal structural identifiability means that it is possible to
determine the parameter value by observing the output for some input and some choice of $\Delta$, but the chosen $\Delta$ need not arise from a scheduling variable.
%For example, if $d=1$, $\Delta$ could be of the form $\Delta=\delta_1 I_{n_2}$, where $\delta_2$ is a transfer function of a stable LTI system. 
It is not \emph{a-priori} clear that these two identifiability notions are equivalent. 
%Note that structural identifiability fo LPV-LFR parametrizations is not a simple consequence of identifiability of the their LTI parts. In fact, 
%LPV-LFR parametrizations can be identifiable without their LTI component being identifiable, see the example below.
%\begin{Example}
%\end{Example}
%using Corollary \ref{col1}, we can show the following.
%\begin{Corollary}[Identifiability of LPV-LFRs]
%\label{ident:theo1}
% An LPV-LFR parametrization is strucural identifiably if and only if it is formaly structuraly identifiable.
%\end{Corollary}
% The significance of Corollary \ref{ident:theo1} is that it tells us that 

In fact, we can show that structural identifiability of ALPV models, and (formal) structural identifiability of LPV-LFRs are equivalent.
To this end, recall from \cite{Alkhoury15} the notions of parametrization, structural identifiability and minimality for ALPVs.
Denote by $\mathcal{LPV}(\NP,\NX,\NU,\NY)$ the set of all ALPV models of the form \eqref{equ:FALPVSystem}. 
An ALPV parametrization is a function $\mathbf{\Sigma}:\Theta \rightarrow \mathcal{LPV}(\NP,\NX,\NU,\NY)$.
An ALPV parametrization $\mathbf{L}$ is \emph{structurally identifiable}, if for any two distinct parameter values $\theta_1, \theta_2 \in \Theta$, $\theta_1 \ne \theta_2$, the input-output maps of $\mathbf{L}(\theta_1)$ and $\mathbf{L}(\theta_2)$ are not equal.
% and $\mathbf{L}$ is \emph{structuraly minimal}, if for all $\theta \in \Theta$, $\mathbf{L}(\theta)$ is minimal. 
We say that an LPV-LFR parametrization $\mathbf{M}$ \emph{originates from an ALPV parametrization $\mathbf{L}$ by MR factorization}, if
 for every $\theta \in \Theta$, $\mathbf{M}(\theta)$ is an LPV-LFR which is calculated from the ALPV
 $\mathbf{L}(\theta)$ by using an MR factorization. 
 Likewise, we say that an ALPV parametrization $\mathbf{L}$ \emph{arises from the LPV-LFR parametrization $\mathbf{M}$}, if for every $\theta \in \Theta$,  $\mathbf{L}(\theta)$ is the ALPV associated with $\mathbf{M}(\theta)$, defined in Definition \ref{LFR2LPV}.
 
 \begin{Theorem}[Identifiability of LPV-LFR and ALPVs]
 \label{ident:theo}
    Consider an ALPV parametrization $\mathbf{L}$ and a LPV-LFR parametrization $\mathbf{M}$.
     \begin{enumerate}
     \item $\mathbf{M}$ is structurally identifiable $\iff$ $\mathbf{M}$ is formally structurally identifiable.
     \item If $\mathbf{M}$ originates from $\mathbf{L}$ by an MR factorization, then,
           $\mathbf{L}$ is structurally identifiable $\iff$ $\mathbf{M}$ is structurally identifiable.
     \item If $\mathbf{L}$ arises from $\mathbf{M}$, then, $\mathbf{L}$ is structurally
           identifiable $\iff$  $\mathbf{M}$ is structurally identifiable. 
     \end{enumerate}
\end{Theorem}
Theorem \ref{ident:theo} implies that in order to identify LPV-LFRs, it is sufficient to identify the 
corresponding ALPVs, and vice versa. In particular, in order to identify LPV-LFR models, it is enough to test them for 
uncertainty blocks  of the form \eqref{DeltaScheduling} which come from scheduling variables.  
Theorem \ref{ident:theo} allows us to use the recent results of \cite{Alkhoury15} to investigate identifiability of LPV-LFRs.

\section{Conclusion}\label{sec:conclusion}
Structural properties of the transformation between ALPV and LFR models are studied. More precisely, minimal, input-output equivalent and identifiable ALPV models are shown to yield minimal, input-output equivalent and identifiable LPV-LFRs respectively, under the condition that the transformation  is performed via a minimal rank factorization. LFR models that can be obtained from ALPV models are characterized using their input-output equivalent input-output maps. In a close future, these equivalence results will allow us to extend system identification solutions for ALPV to LFRs.
%In this paper, we showed that minimal \ALPV models can be converted to minimal \LFPV ones using a minimal rank factorization. We showed that input-output equivalent \ALPV models lead to LFR models which behave in the same manner for any choice of uncertainty block $\Delta$. 
%Indeed, such LFR models have a unique FPS representation because it can be expressed in terms on Markov parameters of the original \ALPV models, which can be uniquely determined from the input-output map of the system.
%we noted out that, certain class of LFR models correspond to \ALPV ones only if their FPS representation has a zero values for certain terms. We also showed that LFR models that correspond to isomorphic \ALPV ones are also isomorphic. \red{REWRITE}

\appendix
\section{Sketches of the proofs}\label{appendix:proofs}    % Each appendix must have a short title.
\begin{proof}[Sketch of the proof of Theorem \ref{th:LFPVMinimality}]
	
	\textit{\textbf {1)}} We will show $\Sigma$ is span-reachable (resp. observable), if and only if $\mathcal{M}$ is
	span-reachable (res. observable).
	%of $\mathcal{R}_{\NX-1}(\Sigma)$ and $\mathcal{O}_{\NX-1}(\Sigma)$ are of full row and column rank.
	We sketch the proof for reachability, observability can be handled in a similar fashion.

	\textit{\textbf{Reachability}}  By induction on $k \in \mathbb{N}$, we can show that
	\begin{equation}
	\label{pf:eq1}
	\begin{split}
	& \IM \mathcal{R}_{k+1}^i(\mathcal{M})=\IM \begin{bmatrix} F_{i,1}\mathcal{R}^1_k(\mathcal{M}), &  G_i \end{bmatrix}, ~ i=2,\ldots,d \\
	& \IM  \mathcal{R}_{k}^1(\mathcal{M}) \subseteq \IM \mathcal{R}_{k}(\Sigma), \\
	& \IM \mathcal{R}_{k}(\Sigma) \subseteq \IM  \mathcal{R}_{2k+1}^1(\mathcal{M}).  \\
	\end{split}
	\end{equation}

	Assume now that $\Sigma$ is span-reachable. Then $\IM \mathcal{R}_{\NX-1}(\Sigma)=\mathbb{R}^{\NX}$, and since by \eqref{pf:eq1},
	$\IM \mathcal{R}_{\NX-1}(\Sigma) \subseteq \IM \mathcal{R}_{2\NX-1}^1(\mathcal{M})$, it follows that
	$\IM \mathcal{R}_{2\NX-1}^1(\mathcal{M})=\mathbb{R}^{\NX}$, \emph{i.e.}, $\rank \mathcal{R}_{2\NX-1}^1(\mathcal{M})=\NX=n_1$.
	Moreover, from the first statement of \eqref{pf:eq1} it follows that
	\[ \mathcal{R}^i_{2\NX} = \begin{bmatrix} F_{1,i}, & G_i \end{bmatrix} \begin{bmatrix} \mathcal{R}_{2\NX-1}^1(\mathcal{M})  & 0 \\ 0 & I_{\NU} \end{bmatrix}. \]
	Since $\rank \begin{bmatrix} F_{1,i}, & G_i \end{bmatrix}=n_i$ , and as $\rank \mathcal{R}_{2\NX-1}^1(\mathcal{M})=\NX=n_1$, 
	$\rank \begin{bmatrix} \mathcal{R}_{2\NX-1}^1(\mathcal{M})  & 0 \\ 0 & I_{\NU} \end{bmatrix}=\NX+\NU$, it follows that
	$\rank \mathcal{R}^i_{2\NX}(\mathcal{M})=n_i$. Moreover, it is easy to see that
	$\rank \mathcal{R}^i_{k}(\mathcal{M}) \le \rank \mathcal{R}^i_{k+1}(\mathcal{M})$ for all $i=1,\ldots,d$. Hence, for
	$k=2\NX$, $\rank \mathcal{R}^i_k(\mathcal{M})=n_i$, $i=1,\ldots,d$, \emph{i.e.}, $\mathcal{M}$ is reachable. 
	
	Conversely, assume that $\mathcal{M}$ is reachable. Then for some $k \ge 0$,
	$\IM \mathcal{R}^1_k(\mathcal{M})=\mathbb{R}^{\NX}$, and from \eqref{pf:eq1} it then follows that 
	$\IM \mathcal{R}_k(\Sigma)=\mathbb{R}^{\NX}$, \emph{i.e.}, $\rank \mathcal{R}_k(\Sigma)=\NX$. Note that in \cite{PM12,PTM15} it was shown that
	$\rank \mathcal{R}_{\NX-1}(\Sigma) \ge \rank \mathcal{R}_k(\Sigma)$ for any $k \ge 0$, and hence 
	$\rank \mathcal{R}_{\NX-1}(\Sigma)=\NX$, \emph{i.e.}, $\Sigma$ is span-reachable.

	\textit{\textbf{2)}}
	Notice that  $Y_{\mathcal{M}}(s)$ equals zero, if $s$ is not a  sequence of the form $i_1\,1\,i_2 \cdots 1\,i_k$,
	$i_1,\ldots,i_k \in \{1,\ldots,k\}$.
	If $s=i_1\,1\,i_2 \cdots 1\,i_k$, then 
	$Y_{\mathcal{M}}(s)=H_{i_k}F_{i_k,1}\dots F_{1,i_1}G_{i_1}=C_{i_k-1}A_{i_{k-1}-1} \cdots A_{i_2-1}B_{i_1-1}$.
	The latter matrix products are precisely the Markov-parameters of $Y_{\Sigma}$ defined in \cite{PM12,PTM15}.
	That is, there is a one-to-one correspondence between $Y_{\mathcal{M}}$ and the Markov-parameters of
	$Y_{\Sigma}$. It was shown in \cite{PM12,PTM15}, that the Markov-parameters of $Y_{\Sigma}$ determine $Y_{\Sigma}$
	uniquely, and vice versa. That is, $\Sigma_1$ and $\Sigma_2$ are input-output equivalent if and only if
	the Markov-parameters of $Y_{\Sigma_1}$ and $Y_{\Sigma_2}$ are the same, and the latter is equivalent to 
	$Y_{\mathcal{M}_1}=Y_{\mathcal{M}_2}$.
	
	%	\begin{align}\label{equ:series_zero_terms_sequence}\dots[i_{k-1}][1][i_{k}][1][i_{k+1}][1]\dots\end{align} 
	%are zeros, \emph{i.e.}, if the term of the series corresponds to a sequence that contain two adjacent $i_t,i_{t-1} $ where $i_t \ne 1$ and $i_{t-1} \ne 1$, then this term is zero. This is simply because we choose the matrix $D_{zw}=0$ while building the LFR-LPV model, and in the canonical partitioning $D_{zw}$ corresponds to the matrices $F_{ij}$ where both $i,j>1$.
	
	%otice also that the terms which correspond to sequences of the mentioned form will be equivalent to Markov parameters of the \ALPV system (see \cite{PM12}), \emph{i.e.}, we will have $C_iB_0$, $C_0B_i$ $C_iA_0B_0$, $C_0A_iB_0$, $C_0A_0B_i$, $C_0A_iA_jB_0$, \emph{etc}., as coefficients of the  formal power series terms, where $A_i, B_i, C_i$ are the matrices of the \ALPV model.
	
	%This connection is very important because, on the one hand, it proves part~2 of Theorem~\ref{th:LFPVMinimality}, \emph{i.e.}, the input-output equivalence between two LFR models that correspond to input-output equivalent \ALPV models, and on the other hand, it characterizes the FPS representations that correspond to \ALPV models as in Theorem~\ref{th:LFR_LPV_equivalence}.
	
	\textit{\textbf{3)}} Left to the reader. 
	%Notice that, if $T$ is and isomorphism between $\Sigma_1$ and $\Sigma_2$, then $T_{\mathcal{M}} = \diag{[T, S]}$ is an isomorphism between $\mathcal{M}_1$ and $\mathcal{M}_2$, where $S$ is a matrix of suitable dimensions.
\end{proof}
\begin{proof}[Sketch of the proof of Theorem \ref{LFR2LPV:theo}]
	\textbf{(i)} 
	It is sufficient to prove that if $\mathcal{M}$ is reachable (resp. observable), then $\Sigma$ is reachable (resp. observable).
	This can be shown in the same way as the corresponding implication in \textbf{(1)} of Theorem \ref{th:LFPVMinimality}.
	More precisely, we can show that \eqref{pf:eq1} holds. Hence, if $\mathcal{M}$ is reachable, then
	$\Sigma$ is span-reachable by the same argument as in the proof of  Theorem \ref{th:LFPVMinimality}. Observability can
	be handled in a similar manner.
	
	\textbf{(ii)} 
	The proof of \textbf{(ii)} is similar to the proof of \textbf{(2)} of Theorem \ref{th:LFPVMinimality}:
	it can be shown that the values of ${Y}_{\mathcal{M}}$ are either zeros or they coincide with the Markov-parameters of
	$Y_{\Sigma}$. Hence, $\mathcal{M}$ and $\mathcal{\tilde{M}}$ are formally input-output equivalent, if and only if 
	the Markov parameters of $\Sigma$ and $\tilde{\Sigma}$ coincide, and by \cite{PM12,PTM15}, the latter is equivalent to
	$\Sigma$ and $\tilde{\Sigma}$ being
	input-output equivalent.
	
	\textbf{(iii)} Easy exercise.
	
	\textbf{(iv)} Note that $\Sigma$ is the ALPV associated with both $\mathcal{M}$ and $\mathcal{\hat{M}}$. Hence, by part \textbf{(ii)}
	$\mathcal{M}$ and $\mathcal{\hat{M}}$ have to be formally input-output equivalent. If $\mathcal{M}$ is minimal, then by
	part \textbf{(i)} so is $\Sigma$, and by Theorem \ref{th:LFPVMinimality}, $\mathcal{\hat{M}}$ is minimal. Since 
	$\mathcal{M}$ and $\mathcal{\hat{M}}$ are both formally input-output equivalent and minimal, then,
	by Theorem \ref{lfr:theo} they are isomorphic. 
	
\end{proof}
\begin{proof}[Proof of Corollary \ref{col2}]
	Let $\mathcal{M}$ be an LPV-LFR, and compute its associated ALPV $\Sigma$. From \cite{PM12,PTM15} it follows that
	$\Sigma$ can be converted to a minimal ALPV $\Sigma_m$ which is input-output equivalent to $\Sigma$. 
	Let us calculate an LPV-LFR $\mathcal{M}_m$ from $\Sigma_m$ using MR factorization. By Theorem \ref{th:LFPVMinimality},
	$\mathcal{M}_m$ is minimal, and by part \textbf{(iv)}, $\mathcal{M}_m$ and $\mathcal{M}$ are formally input-output equivalent.
	Note that $\mathcal{M}_m$ is an LPV-LFR.
	If $\mathcal{M}^{'}$ is any other minimal LFR which is formally input-output equivalent to $\mathcal{M}$, then
	$\mathcal{M}^{'}$ is isomorphic to $\mathcal{M}_m$. It is easy to see that if an LFR satisfies the definition of LPV-LFRs, then
	any LFR which is isomorphic to it will also satisfy the definition of LPV-LFRs.
\end{proof}
\begin{proof}[Proof of Corollary \ref{col1}]
	The corollary is just a reformulation of part \textbf{(ii)} of Theorem \ref{LFR2LPV:theo}.
\end{proof}
\begin{proof}[Proof of Theorem \ref{th:LFR_LPV_equivalence}] 
 It is clear from the definition of an LPV-LFR that if $\mathcal{M}$ is an LPV-LFR, then $Y_{\mathcal{M}}$ satisfies the conditions of the theorem. 
 It is left to show that if $Y_{\mathcal{M}}$ satisfies the conditions of the theorem, then $\mathcal{M}$ is input-output equivalent to an LPV-LFR.
 To this end, we can always assume that $\mathcal{M}$ is minimal, without loss of generality. Otherwise, we can always transform it to an equivalent minimal one. 
We also know that for a minimal LFR, the reachability and observability matrices are full row and column rank respectively.
	We will prove now that if the series has zero terms for the sequences which are not of the form $i_1\,1\,i_2 \cdots 1\, i_k\, 1$, then the matrix blocks $F_{i,j}$ will be zero when both $i,j$ are bigger than $1$, and hence $\mathcal{M}$ is an LPV-LFR. 
	%This structure of LFR corresponds to an \ALPV as it was showed previously in Section~\ref{sec:Connections_minimality_LPV_LFPV}.
	
	For this purpose, we will show that the block matrix $F_{r,q} = 0$, whenever $r > 1$ and $q > 1$. We claim that if $\mathcal{O}_r F_{r,q} \mathcal{C}_q = 0$, where $\mathcal{O}^r=\mathcal{O}^r_{l}(\mathcal{M})$ and $\mathcal{C}_q=\mathcal{R}^q_{l}(\mathcal{M})$ are such that $\rank \mathcal{O}^r=n_r$ and $\rank \mathcal{C}^r=n_r$. Then, $F_{r,q}=0$ due to the fact that $\mathcal{R}^q_{l}(\mathcal{M})$ and $\mathcal{O}^r_{l}(\mathcal{M})$ are full row and column rank respectively.
	
	Let us take the row $H_{i_k}F_{i_k,i_{k-1}}\dots F_{i_1,r}$ of $\mathcal{O}^r_{l}(\mathcal{M})$, and the column $F_{q,j_{k-1}} \dots F_{j_1,j_0}G_{j_0}$ of $\mathcal{R}^q_{l}(\mathcal{M})$, then,
	\begin{align}
	H_{i_k}F_{i_k,i_{k-1}}\dots F_{i_1,r} F_{r,q}F_{q,j_{k-1}} \dots F_{j_1,j_0}G_{j_0}\nonumber \\
	=\mathcal{S}(i_ki_{k-1}\dots i_1\,\,r\,\,q\,\,j_{k-1}\dots j_0)  = 0 .
	\end{align}
	This means that $\mathcal{O}^r_{n-1}(\mathcal{M}) F_{r,q} \mathcal{R}^q_{n-1}(\mathcal{M}) = 0$, \emph{i.e.},  $F_{r,q} = 0$.
\end{proof}
\begin{proof}[Sketch of the proof of Theorem \ref{ident:theo}]
	\textbf{(i)}
	The statement follows by noticing that for any $\theta_1,\theta_2 \in \Theta$,  due to Corollary \ref{col1},
	$\mathbf{M}(\theta_1)$ and $\mathbf{M}(\theta_2)$ are formally input-output equivalent 
	if and only if they are input-output equivalent.
	
	\textbf{(ii)}-\textbf{(iii)}
	Assume that $\mathbf{M}$ originates from $\mathbf{L}$  by MR factorization, or  $\mathbf{L}$ originates from $\mathbf{M}$.
	In both case, for any $\theta_1,\theta_2 \in \Theta$, $\mathbf{L}(\theta_1)$ and $\mathbf{L}(\theta_2)$ are input-output equivalent, if and only if $\mathbf{M}(\theta_1)$ and $\mathbf{M}(\theta_2)$ are input-output equivalent. Both statements of the theorem then follow from the definition of structural identifiability for $\mathbf{M}$ and $\mathbf{L}$. 
\end{proof}
\bibliographystyle{plain}
\bibliography{Bibliography}

\end{document}